\documentclass[]{article}
\usepackage[utf8x]{inputenc}
\usepackage{bm}
\usepackage{amsfonts}
\usepackage{amsthm}
\usepackage{bm}
\newcommand{\bb}{\bm{\beta}}
\usepackage{mathptmx}  
\usepackage{mathptmx}  
\usepackage{graphicx}
\usepackage{mathptmx}
\usepackage[T1]{fontenc}
\usepackage[english]{babel}
\usepackage{amssymb,amsmath,amsfonts}
\usepackage{setspace}
\usepackage{amsthm}
\usepackage{pdfpages}
\usepackage{subfig}

\newtheorem{proposition}{Proposition}
\newtheorem{example}{Example}
\newtheorem{definition}{Definition}
\newtheorem{remark}{Remark}
\newtheorem{corollary}{Corollary}

\DeclareMathOperator*{\argmin}{arg\,min}

\newcommand{\M}{\mathbf{M}}

\newcommand{\y}{\mathbf{y}}

\renewcommand{\P}{\mathbf{P}}

\usepackage[normalem]{ulem}

\title{\vspace{-2cm}Lost in translation: On the impact of data coding on penalized regression with interactions}
\author{Johannes W R Martini$^{1,2}$ \and Francisco Rosales$^{3}$ \and Ngoc-Thuy Ha$^2$ \and Thomas Kneib$^4$ \and  Johannes Heise$^5$ \and  Valentin Wimmer$^1$}

\date{\small $^1$ KWS SAAT SE, Einbeck, Germany \\
  $^2$ Department of Animal Breeding and Genetics, University of Goettingen, Göttingen, Germany \\
$^3$ Department of Finance, Universidad del Pac\'ifico, Lima, Peru \\
$^4$ Chairs of Statistics and Econometrics, University of Goettingen, Göttingen, Germany \\
$^5$ IT Solutions for Animal Production (vit), Verden, Germany}

\begin{document}

\maketitle

\begin{abstract}
Penalized regression approaches are standard tools in quantitative genetics.  
It is known that the fit of an \emph{ordinary least squares} (OLS) regression is independent of certain transformations of the coding of the predictor variables, and that the standard mixed model \emph{ridge regression best linear unbiased prediction} (RRBLUP) is neither affected by translations of the variable coding, nor by global scaling. 
However, it has been reported that an extended version of this mixed model, which incorporates interactions by products of markers as additional predictor variables is affected by translations of the marker coding. In this work, we identify the cause of this loss of invariance in a general context of penalized regression on polynomials in the predictor variables. We show that in most cases, translating the coding of the predictor variables has an impact on effect estimates, with an exception when only the size of the coefficients of monomials of highest total degree are penalized. The invariance of RRBLUP 
can thus be considered as a special case of this setting, with a polynomial of total degree 1, where the size of the fixed effect (total degree 0) is not penalized but all coefficients of monomials of total degree 1 are.
The extended RRBLUP, which includes interactions, is not invariant to translations because it does not only penalize interactions (total degree 2), but also additive effects (total degree 1). Our observations are not restricted to ridge regression, but generally valid for penalized regressions, for instance also for the $\ell_1$ penalty of LASSO.
The fact that coding translations alter the estimates of interaction effects, provides an additional reason for interpreting the biological meaning of estimated interaction effects with caution. Moreover, this problem does not only apply to gene by gene interactions, but also to other types of interactions whose covariance is modeled with Hadamard products of covariance matrices (for instance gene by environment interactions).
\end{abstract}

\newpage 
\section*{Background}
Genomic prediction is the prediction of properties of individuals from their genomic data. It is a crucial ingredient of modern breeding programs \cite{meuwissen01,schaeffer06,habier07,hayes09a,hayes13}.
The traditional quantitative genetics theory is built upon linear models in which allele effects are mostly modeled additively \cite{falconer96}.
In more detail, the standard model to represent the effect of the genotype on the phenotype is given by
\begin{equation} \label{eq:01}
\mathbf{y}=\mathbf{1}_n \mu + \M \bm{\beta} + \bm{\epsilon},
\end{equation} 
where $\mathbf{y}$ is the $n \times 1$ vector of the phenotypic observations of $n$ individuals and $\mathbf{1}_n$ an $n \times 1$ vector with each entry equal to $1$.  Moreover,
$\mu$ is the $y$-intercept, and $\M$  the  $n \times p$ matrix describing the marker states of $n$ individuals at $p$ loci. Dealing with single nucleotide polymorphisms (SNPs) and a diploid species, the entries $M_{i,j}$ can for instance be coded as $0$ $(\mathbf{aa})$, $1$ $(\mathbf{aA}$ or $\mathbf{Aa})$ or $2$ $(\mathbf{AA})$ counting the occurrence of the reference allele $\mathbf{A}$. The $p \times 1$ vector $\bm{\beta}$ represents the allele substitution effects of the $p$ loci, and $\bm{\epsilon}$ the $n \times 1$ error vector. For single marker regression, which may for instance be used in genome-wide association studies (GWAS), we could apply ordinary least squares regression to estimate (or to predict) $\hat{\beta}$. 
However, in approaches of genomic prediction, we model the effects of many different loci simultaneously and the number of markers $p$ is usually much larger than the number of observations $n$. To reduce overfitting and to deal with a large number of predictor variables, different methods have been applied in the last decades, among which \emph{ridge regression best linear unbiased prediction} (RRBLUP) is the most popular \cite{schaeffer2004application}. RRBLUP penalizes the squared $\ell_2$ norm of $\bm{\beta}$ and has been built on the additional model specifications of $\mu$ being a fixed unknown parameter, $\bm{\beta}\sim \mathcal{N}(\bm{0},\sigma_\beta^2 \mathbf{I}_p)$ and $\bm{\epsilon}\sim \mathcal{N}(\bm{0},\sigma_{\bm{\epsilon}}^2 \mathbf{I}_n)$.
With an approach of maximizing a certain density, these assumptions allow to derive the optimal penalty factor as the ratio of the variance components $\lambda:= \frac{\sigma_{\bm{\epsilon}}^2}{\sigma_\beta^2}$ \cite{henderson75,henderson76,henderson77}. 
Please note that RRBLUP is not a pure ridge regression, but actually a mixed model in which the size of $\mu$ is not penalized, but the entries of $\bm{\beta}$ are.  
This mixed model RRBLUP is also called \emph{genomic best linear unbiased prediction} (GBLUP) when it is reformulated with $\mathbf{g}:=\M \bm{\beta}$, and thus $\mathbf{g}\sim \mathcal{N}(0,\sigma_\beta^2 \M' \M)$. \\

It is known that translations of the marker coding, that is, subtracting a constant $m_i$ from the $i$-th column of $\M$, does not change the predictions $\hat{\y}$ of an  OLS regression (provided it is well-defined). This invariance also holds for RRBLUP, when the penalty factor remains fixed. Also when modeling interactions by products of two predictor variables, that is when fitting the coefficients of a polynomial of total degree 2 to the data, OLS predictions are not affected by translations of the marker coding. Contrarily, the predictions of its penalized regression analogue \emph{extended genomic best linear unbiased prediction} (eGBLUP) are sensitive to a translation of the coding \cite{he2016does,Martini17}.\\

In this work, we address the question of why the penalized regression method is affected by translations of the marker coding when a polynomial function of higher total degree is used. 
We  start with a short summary of the different methods.

\section*{Theory: Specification of regression methods}
If an expression includes an inverse of a matrix, we implicitly assume that the matrix is invertible for the respective statement, also if not mentioned explicitly. 
Analogously, some statements for OLS may implicitly assume that a unique estimate exists, which in particular restricts to cases in which the number of observations is at least the same as the number of parameters that have to be determined. \vspace{0.4cm} \\
{\bf Additive effect regression}\vspace{0.2cm}\\ 
The additive effect model has already been presented in Eq.~(\ref{cor:01}). \vspace{0.1cm}\\
{\bf OLS} 
The ordinary least squares approach determines $\hat{\bm{\beta}}$ by minimizing the sum of squared residuals (SSR): 
\begin{equation}\label{eq:03}
\begin{pmatrix}
\hat{\mu} \\
\hat{\bm{\beta}}
\end{pmatrix}_{\mbox{\small OLS}} := \argmin_{(\mu,\bm{\beta}) \in \mathbb{R}^{p+1}} \sum\limits_{i=1}^n (y_i - \M_{i,\bullet}\bm{\beta} - \mu)^2
\end{equation}
$\M_{i,\bullet}$ denotes here the $i$-th row of $\M$ representing the genomic data of individual $i$. The solution to the minimization problem of Eq.~(\ref{eq:03}) is given by the well-known OLS estimate
\begin{equation}\label{eq:04}
\begin{pmatrix}
\hat{\mu} \\
\hat{\bm{\beta}}
\end{pmatrix}_{\mbox{\small OLS}} =
\left( 
\begin{pmatrix}
\mathbf{1}_n & \hspace{-0.2cm}\M
\end{pmatrix}^{t}
\begin{pmatrix}
\mathbf{1}_n & \hspace{-0.2cm} \M
\end{pmatrix} \right)^{-1} 
\begin{pmatrix}
\mathbf{1}_n & \hspace{-0.2cm} \M
\end{pmatrix}^{t}
\mathbf{y}
\end{equation}
provided that the required inverse exists, which in particular also means that $n$ has to be greater than $p$.\\

In problems of statistical genetics, we often deal with a high number of loci and a relatively low number of observations.
In this situation of $p \geq n$, the solution to Eq.~(\ref{eq:03}) is not unique but a vector subspace of which each point minimizes Eq.~(\ref{eq:03}) to zero (``overfitting'').
Using an arbitrary value of this subspace, predictions $\hat{\mathbf{y}}$ for genotypes which have not been used to estimate the parameters $(\hat{\mu},\hat{\beta})$ usually have a low correlation with the corresponding realized phenotypes. An approach to prevent overfitting is RRBLUP.
\vspace{0.4cm}\\
{\bf RRBLUP / GBLUP} minimizes 
\begin{equation}\label{eq:05}
\begin{pmatrix}
\hat{\mu} \\
\hat{\bm{\beta}}
\end{pmatrix}_{\mbox{\small RR}_\lambda} :=\argmin_{(\mu,\bm{\beta}) \in \mathbb{R}^{p+1}} \sum\limits_{i=1}^n (y_i - \M_{i,\bullet}\bm{\beta} - \mu)^2  + \lambda \sum\limits_{j=1}^p {\beta}_j^2
\end{equation}
for a penalty factor $\lambda > 0$. Using an approach of maximizing the density of the joint distribution of $(\mathbf{y},\bm{\beta})$, the model specifications of ${\beta_i}\stackrel{i.i.d.}{\sim} \mathcal{N}(0,\sigma_\beta^2)$ and ${\bm{\epsilon}_i}\stackrel{i.i.d.}{\sim} \mathcal{N}(0,\sigma_{\bm{\epsilon}}^2)$ allow to determine the penalty factor as ratio of the variance components  as $\lambda := \frac{\sigma_{\bm{\epsilon}}^2}{\sigma_\beta^2}$. We stress again that Eq.~(\ref{eq:05}) is not a pure ridge regression, as the name RRBLUP might suggest, but a mixed model which treats $\mu$ and $\bb$ differently by not penalizing the size of $\mu$. This is the version, which is most frequently used in the context of genomic prediction (often with additional fixed effects). \\
The corresponding solution is given by 
\begin{equation}
\begin{pmatrix}
\hat{\mu} \\
\hat{\bm{\beta}}
\end{pmatrix}_{\mbox{\small RR}_\lambda} =
\left( 
\begin{pmatrix}
\mathbf{1}_n & \hspace{-0.2cm}\M
\end{pmatrix}^{t}
\begin{pmatrix}
\mathbf{1}_n & \hspace{-0.2cm} \M
\end{pmatrix} + \lambda \begin{pmatrix}
0 & \mathbf{0 }_p^t \\
\mathbf{0 }_p & \mathbf{I}_p
\end{pmatrix} 
\right)^{-1} 
\begin{pmatrix}
\mathbf{1}_n & \hspace{-0.2cm} \M
\end{pmatrix}^{t}
\mathbf{y}. \label{eq:06}
\end{equation}
where $\mathbf{0 }_p$ denotes the $p \times 1$ vector of zeros. The effect of the introduction of the penalization term $\lambda \sum\limits_{j=1}^p {\beta}_j^2$ is that for the minimization of Eq.~(\ref{eq:05}), we have a trade-off between fitting the data optimally and shrinking the squared effects to $0$. The method will only ``decide'' to increase the estimate $\hat{\beta}_j$, if the gain from improving the fit is greater than the penalized loss generated by the increase of $\hat{\beta}_j$.  \vspace{0.4cm} \\
{\bf First order epistasis: Polynomials of total degree two}\vspace{0.2cm}\\ 
An extension of the additive model of Eq.~(\ref{eq:01}) is a first order epistasis model given by a polynomial of total degree two in the marker data \cite{ober15,jiang15,Martini16}
\begin{equation} \label{eq:07}
y_i =\mathbf{1}_n \mu + \M_{i,\bullet}\bm{\beta} + \sum\limits_{k=1}^p\sum\limits_{j=k+1}^p h_{j,k}M_{i,j}M_{i,k} + \bm{\epsilon},
\end{equation}
Please note that there is a variant of this model, in which also $j=k$ is included. This interaction of a locus with itself allows to model dominance \cite{Martini16}. \\

We recapitulate some terms which are important in the context of polynomials in multiple variables. Each product of the variables $M_{i,1},M_{i,2},\dots,M_{i,p}$ is called a monomial.
For instance $M_{i,1}$, $M_{i,2}$, $M_{i,1} M_{i,2}$ and $M_{i,1}^2$ are four different monomials. Since the product is commutative, $M_{i,1} M_{i,2}$ and $M_{i,2} M_{i,1}$ are the same monomial (and their coefficients are assumed to be summed up in any polynomial which we will address later).
Moreover, the total degree of a monomial is the sum of the powers of the variables in the respective monomial. 
For instance, $M_{i,1}$ and $M_{i,2}$ are monomials of total degree 1, whereas $M_{i,1}$$M_{i,2}$, and $M_{i,1}^2$ are monomials of total degree 2.
Moreover, $M_{i,1}M_{i,2}$ is a monomial of degree 1 in each variable $M_{i,1}$ and $M_{i,2}$ and $M_{i,1}^2$ is a monomial of degree 2 in $M_{i,1}$ and of degree $0$ in $M_{i,2}$. 
Since a polynomial model is also linear in the coefficients, the regression equations are only slightly modified.\\

{\bf OLS} Eq.~(\ref{eq:04}) with a modified matrix $\M$ including the products of markers as additional predictor variables represents the OLS solver for model~(\ref{eq:07}). \vspace{0.4cm}\\
{\bf eRRBLUP}
The extended RRBLUP is based on Eq.~(\ref{eq:07}) and the assumptions of $\mu$ being fixed, ${\beta_i}\stackrel{i.i.d.}{\sim} \mathcal{N}(0,\sigma_\beta^2)$, $h_{j,k }\stackrel{i.i.d.}{\sim} \mathcal{N}(0,\sigma_h^2)$ and ${{\bm{\epsilon}}_i}\stackrel{i.i.d.}{\sim} \mathcal{N}(0,\sigma_{\bm{\epsilon}}^2)$. In this case, the solution is also given by an analogue of Eq.~(\ref{eq:06}), but with two different penalty factors, $\lambda_1:=\frac{\sigma_{\bm{\epsilon}}^2}{\sigma_\beta^2}$ for additive effects and $\lambda_2:=\frac{\sigma_{\bm{\epsilon}}^2}{\sigma_h^2}$ for interaction effects.
\vspace{0.4cm}\\
{\bf Translations of the marker coding}\vspace{0.2cm}\\
In quantitative genetics, column means are often subtracted from the original $0$, $1$, $2$ coding of $\M$ to 
use $\tilde{\M}:=\M-\mathbf{1}_n\P^t$ with $\P$ the vector of column means of $\M$ \cite{vanraden2008efficient} such that 
$$\sum_{i=1}^n \tilde{M}_{i,j} = 0\quad \forall j =1,\dots,p.$$ 
However, other types of translations, for instance a symmetric $\{-1,0,1\}$ coding or a genotype-frequency centered coding \cite{alvarez2007unified,vitezica} can also be found in quantitative genetics' literature. Thus, the question occurs whether this has an impact on the estimates of the marker effects or on the prediction of genetic values of genotypes. \\

The answer is that for the additive setup of Eq.~(\ref{eq:01}), a shift from $\M$ to $\tilde{\M}$ will change $\hat{\mu}$ but not $\hat{\bm{\beta}}$ and any prediction $\hat{\mathbf{y}}$ will not be affected, neither for OLS, nor for RRBLUP (provided that $\lambda$ is not changed). This invariance of the additive model does not hold for the extended RRBLUP method. \\

We give an example and discuss the effect of translations of the marker coding in a more general way afterwards. 
\begin{example}[Translations of the marker coding]\label{ex:06}
	Let the marker data of five individuals with two markers be given:
	$$\mathbf{y}= (-0.72,2.34,0.08,-0.89,0.86)^t \qquad 
	\M = \begin{pmatrix}
	2 & 2 \\
	1 & 2 \\
	2 & 0 \\
	2 & 1 \\
	1 & 0 \\
	\end{pmatrix} $$ 
	Moreover, let us use the original matrix $\M$, and the column mean centered matrix $\tilde{\M}:= \M - \mathbf{1}_5 \underbrace{(1.6,1.0)}_{=:\P^t}$.
	We consider the first order epistasis model 
	$$y_i := \mu + \beta_1 M_{i,1} + \beta_2 M_{i,2}  + h_{1,2} M_{i,1}M_{i,2} + \bm{\epsilon}_i$$
	and estimate the corresponding parameters with i) an OLS regression, ii) a mixed model regression eRRBLUP-1 with $\lambda_1=\lambda_2=1$, and iii) a mixed model regression eRRBLUP-2 with $\lambda_1=0$ and $\lambda_2=1$. The difference between eRRBLUP-1 and eRRBLUP-2 is that the first method penalizes the additive effects and the interaction effect, whereas the latter method only penalizes the interaction effect. The results are reported in Table~\ref{table:example1}. We summarize our observations from the reported results as follows:
	\begin{itemize}	
		\item Comparing the centered and non-centered versions of OLS, the estimates for $\mu$, $\beta_1$ and $\beta_2$ change, but the estimated interaction $\hat{h}_{1,2}$ as well as the prediction of ${\mathbf{y}}$ remains unchanged.
		\item  Comparing the centered and non-centered versions of eRRBLUP-1, both codings give different estimates for all the parameters and these solutions produce different predictions for ${\mathbf{y}}$.
		\item Comparing the centered and non-centered versions of eRRBLUP-2, both codings give different estimates for $\mu$, $\beta_1$ and $\beta_2$, but the same for $h_{1,2}$ and the same predictions for $\y$.
	\end{itemize}
\end{example}
	\begin{table}[ht] \caption{Results from Example 1. ``nc'' denotes the use of the non-centered matrix $\M$ and ``c'' indicates the use of the centered matrix  $\tilde\M$.}\label{table:example1}
	\centering
	\begin{tabular}{|l|r|r|r|r|r|r|}
		\hline
		&\multicolumn{2}{c|}{OLS}&\multicolumn{2}{c|}{eRRBLUP-1} &\multicolumn{2}{c|}{eRRBLUP-2}\\ 
		\hline
		Estimates &nc&c &nc &c &nc&c\\
		\hline
		$\hat{\mu}$&1.83& 0.33& 1.81& 0.33& 2.69 & 0.33\\
		$\hat{\beta_1}$&-0.97 & -2.11& -0.89& -1.15 & -1.54 & -2.11\\
		$\hat{\beta_2}$&1.88 & 0.06 & 0.71& 0.09& 1.03& 0.11\\
		$\hat{h}_{1,2}$&-1.14& -1.14 & -0.48 & -0.57&-0.57 & -0.57\\
		\hline\hline
		$\hat{{\mathbf{y}}}$ &&&&&&\\
		& -0.91& -0.91 & -0.46& -0.27& -0.63& -0.63 \\
		& 2.34& 2.34& 1.39 &1.46&2.06& 2.06\\
		& -0.11& -0.11& 0.03& 0.01&-0.40& -0.40\\
		& -0.51& -0.51& -0.21& -0.13&-0.51& -0.51\\
		& 0.86& 0.86 & 0.92 & 0.59&1.15& 1.15\\              
		\hline 
	\end{tabular}
\end{table}

The different cases presented in Example~\ref{ex:06} have a certain systematic pattern, which we  discuss in the following section. 

\newpage
\section*{Results} 
The observations made in Example~\ref{ex:06} are explained by the following proposition
which has several interesting implications. 
\begin{proposition}\label{prop:01} Let $\M_{i,\bullet}$ be the $p$ vector of the marker values of individual $i$ and let $f(\M_{i,\bullet}): \mathbb{R}^p \rightarrow \mathbb{R}$ be a polynomial of total degree $D$ in the marker data. Moreover, let $\tilde{\M}:= \M - \mathbf{1}_n \P^t$ be a translation of the marker coding and let us define a polynomial $\tilde{f}$ in the translated variables $\tilde{\M}$ by $\tilde{f}(\tilde{\M}_{i,\bullet}):= f(\tilde{\M}_{i,\bullet} +  \P^t)=f(\M_{i,\bullet})$. Then for any data $\y$, the sum of squared residuals (SSR) will be identical 
	$$ \sum_{i=1,...,n} (y_i - f(\M_{i,\bullet}))^2 = \sum_{i=1,...,n} (y_i - \tilde{f}(\tilde{\M}_{i,\bullet}))^2 $$
	and for any monomial $m$ of highest total degree $D$, the corresponding coefficient $a_m$ of $f(\M_{i,\bullet})$ and $\tilde{a}_m$ of $\tilde{f}(\tilde{\M}_{i,\bullet})$ will be identical:
	$$a_m = \tilde{a}_m.$$ 
\end{proposition}
\begin{proof}
	The fact that the SSR remains the same, results from the definition of the polynomials. To see that the coefficients of monomials of highest total degree D are identical, choose a monomial $m(M_{l_1},M_{l_2},...,M_{l_D})$ of the loci $l_1,...,l_D$ of total degree $D$ of $f$. Multiplying the factors of $m(\tilde{M}_{l_1}+P_{l_1},\tilde{M}_{l_2}+P_{l_2},...,\tilde{M}_{l_D}+P_{l_D})$ gives the same monomial 
	$m(\tilde{M}_{l_1},\tilde{M}_{l_2},...,\tilde{M}_{l_D})$ as a summand of highest total degree, plus additional monomials of lower total degree. Thus, the coefficients of monomials of total degree $D$ remain the same.
\end{proof}
Proposition~\ref{prop:01} has the very simple statement that if we have a certain fit $f$ based on $\M$, and we use the translated marker coding $\tilde{\M}$, the polynomial $\tilde{f}$ will fit the data with the same SSR and with the same predictions $\hat{\y}$ (due to the definition of $\tilde{f}$). Moreover, the coefficients of monomials of highest total degree will be the same.  \\

Since OLS is defined only by the minimal SSR, this also means that it is invariant to any translation of the coding, provided that $\tilde{f}$ of Proposition~\ref{prop:01} is a valid fit. To make sure that $\tilde{f}$ is a valid fit, the possibility to adapt coefficients of monomials of lower total degrees is required.
We cannot adapt the regression completely if certain coefficients are forced to zero by the model structure. If a coefficient is equal to zero in $f$, it may be different from zero in $\tilde{f}$.
We illustrate this with an example. \newpage 

\begin{example}[Models without certain terms of intermediate total degree]\label{ex:07}
	Let us consider the data $\M$ and $\y$ of Example \ref{ex:06} but with the assumption that marker $2$ does not have an additive effect. Then
	$$\begin{pmatrix}
	\hat{\mu} \\
	\hat{\beta}_1 \\
	\hat{h}_{1,2}
	\end{pmatrix}_{\mbox{\small OLS}} = \begin{pmatrix}
	3.71 \\
	-2.098\\
	-0.012\\
	\end{pmatrix}
	\qquad \mbox{and} \qquad \begin{pmatrix}
	\tilde{\mu} \\
	\tilde{\beta}_1 \\
	\tilde{h}_{1,2}
	\end{pmatrix}_{\mbox{\small OLS}}= \begin{pmatrix}
	0.334\\
	-2.11 \\
	-1.162
	\end{pmatrix} $$
	and also the estimates $\hat{\y}$ and $\tilde{\y}$ are different. 
\end{example}
Example~\ref{ex:07} illustrates that ``completeness" of the model is required to have the possibility to adapt to translations of the coding. We define this property more precisely.
\begin{definition}[Completeness of a polynomial model]\label{def01}
 Let $\M_{i,\bullet}$ be the $p$ vector of the marker values of individual $i$ and let $f(\M_{i,\bullet}): \mathbb{R}^p \rightarrow \mathbb{R}$ be a polynomial of total degree $D$ in the marker data. The polynomial model $f$ is called complete if for any monomial $\M_{i,j_1}^{d_1} \M_{i,j_2}^{d_2}\cdot \cdot  \cdot \M_{i,j_m}^{d_m}$ of $f$, all monomials 
 $$\M_{i,j_1}^{\delta_1} \M_{i,j_2}^{\delta_2}\cdot \cdot  \cdot  \M_{i,j_m}^{\delta_m} \qquad \forall \; 0 \leq {\delta_1} \leq d_1, \; \forall \; 0 \leq {\delta_2} \leq d_2, \; ...\; , \forall \; 0 \leq {\delta_m} \leq d_m$$ 
 are included with an coefficient to be estimated.
	\end{definition}
Although, this definition seems rather abstract, its meaning can be understood easily by an example. 
Let us consider Eq.~(\ref{eq:07}). Its monomials are of shape $M_{i,k}$ or $M_{i,k}M_{i,l}$. For $M_{i,k}$, Definition~\ref{def01} states that $M_{i,k}^0=1$ and $M_{i,k}^1$ have to be included, which is obviously the case. For $M_{i,k}M_{i,l}$, $M_{i,k}^0=1$, $M_{i,k}^1$ and  $M_{i,k}^1M_{i,j}^1$ have to be included, which is also true. Thus, the model is complete. Analogously, if we also include the interactions $M_{i,k}^2$, the model remains complete. Contrarily, Example~\ref{ex:07} is based on the model 
$$ y_i = \mu+ \beta_1 M_{i,1} + h_{1,2}M_{i,1}M_{i,2} + \epsilon_i.$$
Since $M_{i,1}M_{i,2}$ is included with a coefficient to be estimated, $M_{i,1}$ and $M_{i,2}$ have to be included to make the model complete. Since $M_{i,2}$ is not included, the polynomial is not complete.  \\

Given that the model is complete, Proposition~\ref{prop:01} has various implications.
The following corollaries explain the results observed in our examples and highlight some additional properties of penalized regression methods in general.  
For all statements, it is assumed that penalty factors remain unchanged and that the model is complete.

\begin{corollary}\label{cor:01}
	For a model of any total degree $D$, the OLS estimates of the coefficients of highest total degree as well as the predictions $\hat{\y}$ are invariant with respect to translations of the marker coding. 
\end{corollary}
Corollary~\ref{cor:01} is a result of the OLS method being defined only by the SSR, and $f$ and the corresponding $\tilde{f}$ of Proposition~\ref{prop:01} fitting the data with the same SSR when their respective coding is used. The statement of Corrolary~\ref{cor:01} has been observed in Example~\ref{ex:06}, where the OLS fits for $\hat{\y}$ are identical when the coding is translated, and where the estimated coefficients $\hat{h}_{1,2}$ of highest total degree remain unchanged. 

\begin{corollary}\label{cor:02} For a polynomial model of total degree $D$, and a penalized regression which only penalizes the coefficients of monomials of total degree $D$, the estimates of the coefficients of monomials of total degree $D$, as well as the predictions $\hat{\y}$ are invariant with respect to translations of the marker coding. 
\end{corollary}

Corollary~\ref{cor:02} is a result of the following observation: for each $f$, its corresponding $\tilde{f}$ will have the same SSR (each polynomial with its respective coding), and the same coefficients of highest total degree. Thus, it will have the same value for the target function which we aim to minimize (The target function is the analogue of Eq.~(\ref{eq:05}) with a penalty on only the coefficients of monomials of highest total degree).
Because this is true for any polynomial $f$, it is in particular true for the solution minimizing the target function. 
A central point of Corollary~\ref{cor:02} is that it is valid for any penalty on the size of the estimated coefficients of highest total degree. The sufficient condition is that only these coefficients of highest total degree are penalized.

\begin{corollary}\label{cor:03}	RRBLUP predictions $\hat{\y}$ are invariant with respect to translations of the marker coding. 
\end{corollary}
Corollary~\ref{cor:02} applied to complete models of total degree $1$ gives the result of Corollary~\ref{cor:03}, that is RRBLUP being invariant to translations of the marker coding.
This fact has been previously proven using a marginal likelihood setup \cite{stranden11}, or the mixed model equations \cite{Martini17}.

\begin{corollary}\label{cor:04} An additive least absolute shrinkage and selection operator (LASSO) regression \cite{tibshirani96} based on a polynomial model of total degree 1 and $\ell_1$~penalizing the additive marker effects but not the intercept, is invariant to translations of the marker coding. 
\end{corollary}
Corollary~\ref{cor:04} is a special case of Corollary~\ref{cor:02}. \\

We give a small example, highlighting cases which are not invariant to translations of the marker coding.
We recommend to use the data of Example~\ref{ex:06} to validate the statements.\\

\begin{example}[Regressions affected by marker coding]
	\begin{itemize} \item[a)] Pure ridge regression of an additive model of Eq.~(\ref{eq:01}) with a penalty on the size of $\mu$ (``random intercept'') is not invariant to translations. 
		\item[b)] RRBLUP with the fixed effect forced to zero is not invariant to translations of the marker coding.
		\item[c)] An extended LASSO $\ell_1$ penalizing additive effects and interactions is not invariant to translations of the coding.
	\end{itemize}
\end{example} 

\begin{remark}
	Proposition~\ref{prop:01} stated that the coefficients of monomials of highest total degree $D$ of $f$ and $\tilde{f}$ will be identical. This statement can even be generalized for some situations. Consider for instance the model 
	$$ y_i = f(M_{i,1},M_{i,2},M_{i,3}) = \mu + \beta_1 M_{i,1} + \beta_2 M_{i,2} + \beta_3 M_{i,3} + h_{2,3} M_{i,2}M_{i,3} + \epsilon_i$$
	The model is a polynomial $f$ of total degree 2. Thus, Proposition~\ref{prop:01} states that the coefficient of monomial $M_{i,2}M_{i,3}$ will be identical for $f$ and $\tilde{f}$.
	However, since $M_{i,1}$ is not included in any other monomial, its coefficient will also be identical for both polynomials. We did not generalize Proposition~\ref{prop:01} into this direction to make the manuscript not more technical than necessary. The statement we made in Proposition~\ref{prop:01} is sufficient to explain the observations related to genomic prediction models. 
	\end{remark}

\section*{Discussion}
The illustrated problem of the coding having an impact on the estimates of interactions in penalized regressions is essential for quantitative genetics, where Hadamard products are often used to model interaction such as epistasis or gene by environment interaction \cite{perez17}. Hadamard products of covariance matrices represent exact reformulations of certain interaction effect models \cite{jiang15,Martini16}. In particular, our observations illustrate once more that the size of interaction effect estimates should be interpreted with caution because a biological meaning is not necessarily given. \\ 

It should be highlighted, that the problem does not seem to be a consequence of non-orthogonality of the predictor variables (marker values and their products), since these problems would not appear in an OLS regression (provided that a unique solution exists), where the variables have the same coding and thus the same angle. \\

Finally, note that it has been reported that a Gaussian reproducing kernel regression \cite{Morota14} can be interpreted as a limit of a polynomial regression with increasing total degree (and all possible monomials) \cite{jiang15}. 
Being a limit case of a method which is affected by translations of the coding, the question appears why the Gaussian kernel regression is invariant to translations of the marker coding. It may be interesting to reconsider the limit behavior from a theoretical point of view.

\section*{Conclusion}
We identified the cause of the coding-dependent performance of epistasis effects models. Our results were motivated by ridge regression, but do equally hold for many other types of penalized regressions, for instance for the $\ell_1$ penalized LASSO. The fact that the estimated effect sizes depend on the coding highlights once more that estimated interaction effect sizes should be interpreted with caution with regard to their biological, mechanistic meaning. Moreover, the problem of coding is not only present for marker by marker interaction, but for any mixed model in which interactions are modeled by Hadamard products of covariance matrices, in particular also for gene by environment (G x E) models.


\section*{Author's contribution}
JWRM: Proposed to consider the topic; derived the theoretical results; JWRM and FR wrote the manuscript;
All authors: Discussed the research

\section*{Acknowledgements}
JWRM thanks KWS SAAT SE for financial support during his PhD. FR thanks Universidad del Pac\'ifico for financial support. 
\end{document}